\documentclass[review,authoryear]{elsarticle}

\usepackage{amssymb}
 \usepackage{amsthm}

\newtheorem{theorem}{Theorem}
\newtheorem{lemma}[theorem]{Lemma}

\newcommand{\NNI}{\mathrm{NNI}}

\newcommand{\cC}{{\mathcal C}}

\newcommand{\cN}{{\mathcal N}}

\newcommand{\tc}{\ensuremath{\Delta}} 
\newcommand{\tcp}{\ensuremath{\Delta^{+}}} 
\newcommand{\tcm}{\ensuremath{\Delta^{-}}} 

\begin{document}

\begin{frontmatter}



\title{Transforming phylogenetic networks: Moving beyond tree space}


\author[uea]{Katharina T. Huber}
\ead{katharina.huber@cmp.uea.ac.uk}

\author[uea]{Vincent Moulton}
\ead{vincent.moulton@cmp.uea.ac.uk}

\author[uea]{Taoyang Wu}
\ead{taoyang.wu@gmail.com}
\address[uea]{
School of Computing Sciences,University of East Anglia, \\ Norwich,
NR4 7TJ, UK}



\begin{abstract}

Phylogenetic networks are a  generalization of phylogenetic trees that are   used to represent reticulate evolution. Unrooted phylogenetic networks form a special class of such networks, which naturally generalize unrooted phylogenetic trees. In this paper we define two operations on unrooted phylogenetic networks, one of which is a generalization of the well-known nearest-neighbor interchange (NNI) operation on phylogenetic trees. We show that any unrooted phylogenetic network can be transformed into any other such network using only these operations. This generalizes the well-known fact that any phylogenetic tree can be transformed into any other such tree using only NNI operations. It also allows us to define a generalization of tree space and to define some new metrics on unrooted phylogenetic networks. To prove our main results, we employ some fascinating new connections between phylogenetic networks  and cubic graphs that we have recently discovered. Our results should be useful in developing new strategies to search for optimal phylogenetic networks, a topic that has recently generated some interest in the literature, as well as for providing new ways to compare networks.

\end{abstract}


\end{frontmatter}


\section{Introduction}

Phylogenetic networks are a
generalization of phylogenetic trees
that are gaining growing acceptance by biologists due to their
importance in representing reticulate evolution
\citep{bap-13}. Certain types of 
networks, such as neighbornets~\citep[see e.g.][]{bryant2004neighbor, huson2006application} and 
median networks~\citep[see e.g.][]{bandelt1995mitochondrial} 
are now commonly used in the literature.
Moreover, there has recently been much focus on 
developing ways to construct special classes of networks
to explicitly model evolution~\citep[see e.g.][]{hus-10,nak-11,gus-14}.
Even so, there are still several aspects of phylogenetic network theory 
that remain to be more fully explored. One
such aspect is how to 
transform one network into another one by using 
a collection of specified network 
operations~\citep[see e.g.][for some 
results in this direction]{cardona2009metrics, yu2014maximum, hub-15}
which we consider in this paper.

\def \eleSize{\fontsize{20}{20.4}}
\begin{figure*}[h]
   \begin{center}
  {\resizebox{1\columnwidth}{!}{{\includegraphics{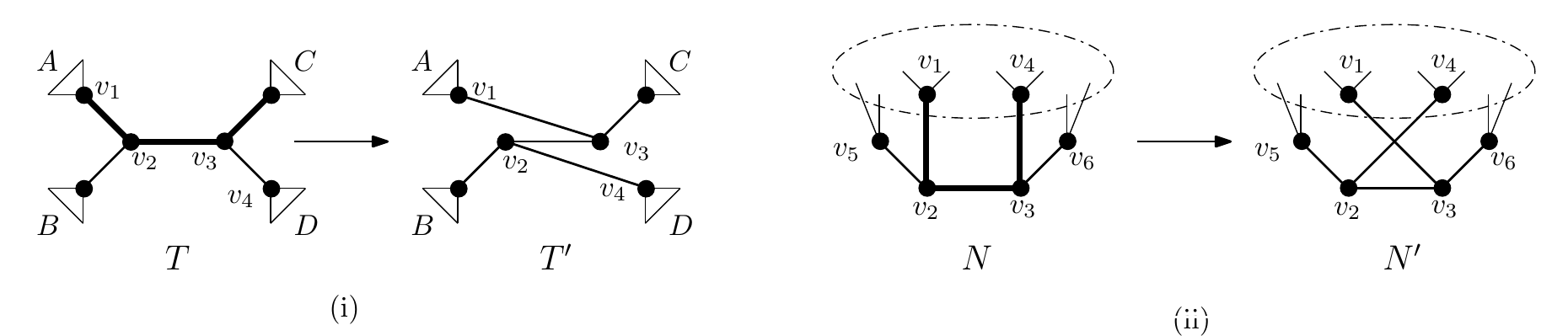}}}  }
  \end{center}
\caption{ (i) An NNI operation on a phylogenetic tree $T$. The tree $T'$ is obtained from $T$ by performing an operation on the path highlighted in bold that results in the subtrees labelled $A$ and $D$ being swapped. (ii) An NNI operation on a network $N$. The network $N'$ is obtained from $N$ by one NNI operation that is performed on the path highlighted in bold, just like for phylogenetic trees. Note that vertices $v_1$, $v_4$, $v_5$ and $v_6$ in $N$ could all have degree 1 or 3, and that $N$ contains neither edge $\{v_1, v_3\}$ nor edge $\{v_2,v_4\}$.
\label{fig:nni}}
\end{figure*}

Transformations of phylogenetic trees 
have been studied for several years,
and have applications to tree search algorithms 
and comparing trees~\citep[cf. e.g. ][Chapters 4 and 30]{felsenstein2004}.
Probably the best known and simplest 
way to transform one phylogenetic tree 
into another is to use a 
{\em nearest-neighbor interchange (NNI) operation}
which we now recall.
For a set $X$ of three or more species or taxa, 
a {\em phylogenetic tree (on $X$)}
is a tree in which every vertex has
degree 1 or 3 with leaf set $X$. A pair
of distinct trees differ by one NNI operation
if one tree can be obtained from the other by swapping 
two of the four subtrees adjacent with an 
interior edge (Figure~\ref{fig:nni}). 
Note that the NNI operation 
is {\em reversible}, i.e. there is a unique NNI 
operation (or {\em reverse} operation) 
that can be applied to get back to the 
original tree. 
A well-known result concerning the NNI operation 
states that given any pair of 
phylogenetic trees $T, T'$ on $X$, 
it is possible to transform $T$ into $T'$ 
by some sequence of NNI operations \citep{rob-71}. 
This implies that NNI operations 
can be used to explore all possible 
phylogenetic trees on a set $X$, 
a useful fact that underpins several 
algorithms for reconstructing phylogenetic trees. 

We now turn to the analogous problem for phylogenetic networks.
More specifically, we shall consider the problem 
of transforming {\em unrooted phylogenetic networks (on X)}, or
{\em networks}, into one another. Such networks are 
connected graphs in which every vertex has
degree 1 or 3 and whose leaf set is $X$ \citep{gam-ber-12}, and
 they have been used to model genome fusion \citep{3d-parsimony}.  
In Figure~\ref{fig:nni:real}, we present 
an example of such a network $R$ that
was referred to as the ``ring of life'' in
a study concerning the genome fusion origin of eukaryotes
\citep[Figure 1]{riv-04}. 
Based on that network, whose
construction employed whole genomes spanning
the diversity of life, the authors concluded that the eukaryotic nuclear
genome has resulted from a fusion of a relative of a protobacterium
($P_{\gamma}$) and a relative of an archeal ecocyte (E).
Note that if a network is a tree, then it is necessarily
a phylogenetic tree on $X$, and so 
networks generalize phylogenetic trees. 
Properties of networks have been studied in \cite{gam-ber-12}, 
and they can be generated by software such
as T-REX~\citep[cf.][]{makarenkov2001t} and 
Splitstree~\citep[cf.][]{huson2006application, huson2005reconstruction}. 
In this paper, we show
that it is possible to extend results concerning 
NNI operations on phylogenetic trees in 
a natural way to networks, which we expect could
lead to applications to network search algorithms
and comparison of networks.

 \def \eleSize{\fontsize{20}{20.4}}
\begin{figure*}[h!]
    \begin{center}
   {\resizebox{0.8\columnwidth}{!}{{\includegraphics{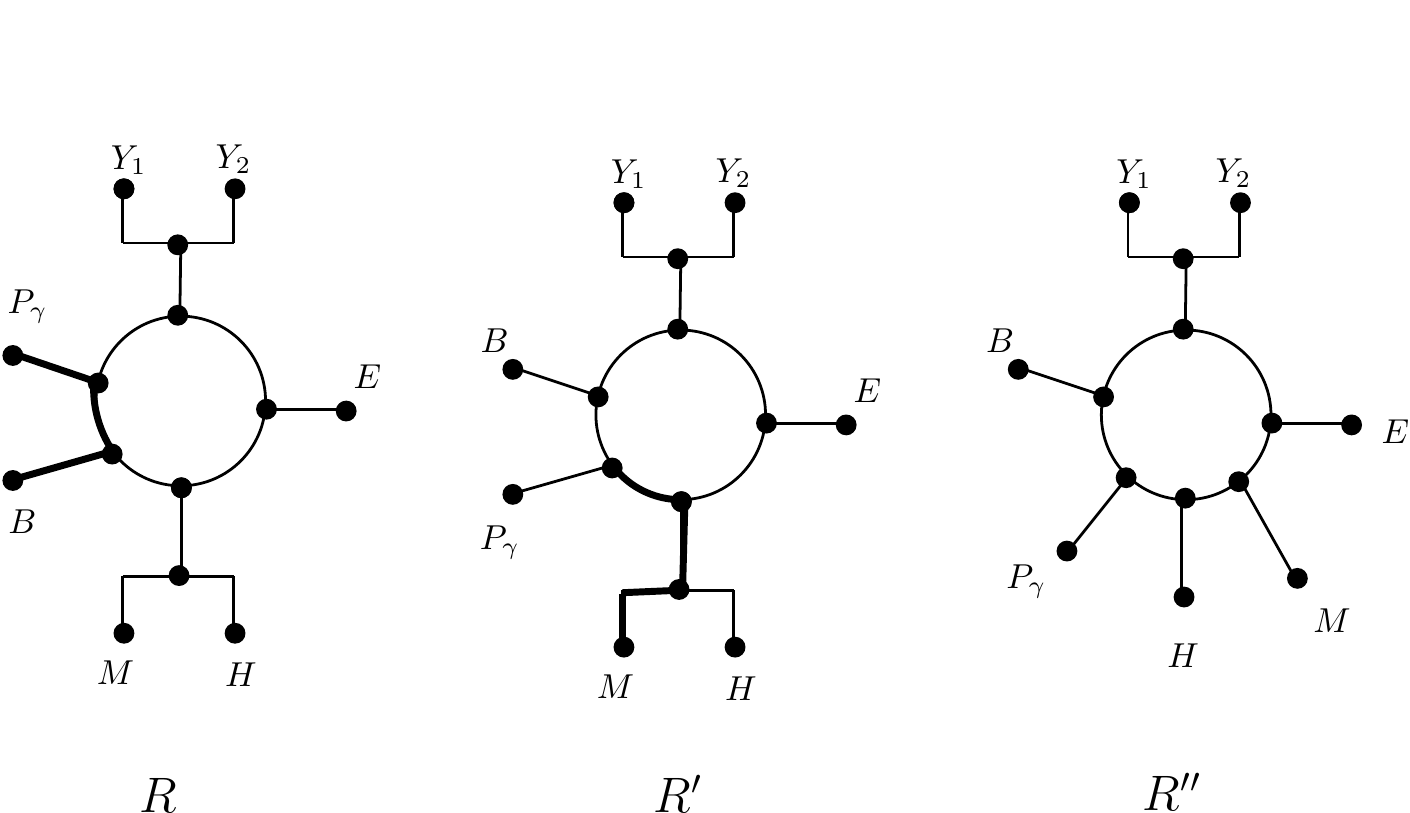}}} }
  \end{center}
\caption{An example of transforming phylogenetic networks using NNI operations. Network $R$ is from \citet[Figure 1]{riv-04}, and is labelled by two yeasts  ($Y_1$, Schizosaccharomyces pombe and $Y_2$, Saccharomyces cerevisiae), a $\gamma$-proteobacterium  ($P_{\gamma}$, Xylella fastidiosa), a bacillus ($B$, Staphylococcus aureus MW2), a halobacterium ($H$, Halobacterium sp. NRC-1), an eocyte ($E$, Sulfolobus tokodaii) and a methanococcus ($M$, Methanosarcina mazei Goe1). The networks $R'$ and $R''$ are obtained from $R$  and $R'$, respectively, by applying an NNI operation to the path highlighted in bold. }
\label{fig:nni:real}
\end{figure*}

Our starting point is to extend NNI operations 
to networks. This is based 
on the observation that the tree $T'$ in Figure~\ref{fig:nni}  
can be obtained from $T$ by replacing the length 3 path 
$v_1,v_2,v_3,v_4$ highlighted in bold
with the path $v_1,v_3,v_2,v_4$ 
whilst preserving all other edges. 
This definition immediately extends to give what we 
shall call an NNI operation
on networks. In particular, we start again with a path 
$v_1,v_2,v_3,v_4$ in a network $N$ on $X$ for which 
neither $\{v_1,v_3\}$ nor $\{v_2,v_4\}$ is an edge, 
and obtain a new network $N'$ on $X$ by
replacing this path with the path $v_1,v_3,v_2,v_4$
(Figure~\ref{fig:nni}).
Note that $N'$ has the same number of vertices
as the original network, and that, just as 
with phylogenetic trees, the NNI operation is reversible.

In the first of our main results we  show
that, just as with phylogenetic trees, we can transform
any network $N$ on $X$ to any other network
$N'$ on $X$ with the same number of
vertices as $N$ by just using NNI operations
(Theorem~\ref{transform2}). 
We illustrate this in Figure~\ref{fig:nni:real}:
here network $R$ is
transformed into the network $R''$ --
which can be considered as an
alternative ring of life hypothesis -- by applying 
a sequence of two NNI operations. 
Even so, it is not possible to explore all 
alternative ways to potentially
represent the ring of life 
using only NNI operations.
To see this, note that we cannot
transform the tree $T$
in Figure~\ref{fig:ring} to the network $R$ 
in Figure~\ref{fig:nni:real} using only NNI operations
since, even though these networks are on the same leaf set, 
they have a different number of vertices
and NNI operations must preserve this number.

\def \eleSize{\fontsize{20}{20.4}}
\begin{figure*}[h]
    \begin{center}
    {\resizebox{0.5\columnwidth}{!}{ {\includegraphics{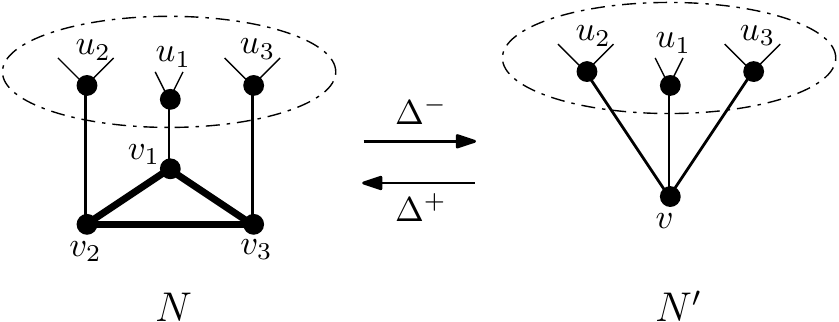}}}  }
  \end{center}
\caption{The two considered triangle operations:   $N'$ is obtained from $N$ by one \tcm~operation that collapses the triangle $v_1,v_2,v_3$  highlighted in bold to the vertex $v$, preserving all other edges that are not contained in this triangle. Conversely, $N$ is obtained from $N'$ by one \tcp~operation which replaces the vertex $v$ by the triangle $v_1,v_2,v_3$. Note that in both networks vertices $u_1, u_2,u_3$  are all distinct and each of them could have degree 1 or 3. }
\label{fig:twoops}
\end{figure*}

In our second main result (Theorem~\ref{transform1})
we show that we need only one additional
operation and its reverse to be able
to transform any network on $X$ into any other network on $X$. 
We call these operations {\em triangle operations} or {\em \tc\, operations}
and picture them in Figure~\ref{fig:twoops}; they involve 
either inserting or removing a triangle (or, more
technically, a length three cycle) 
from a network. We illustrate this 
result in Figure~\ref{fig:ring}.
The tree $T$ in this figure is
one of the phylogenetic trees presented in  
\citet[Figure 1]{riv-04}. 
It is one of the five most probable trees
 computed using the method of conditioned reconstruction.
The ring of life network 
$R$ can be obtained from $T$ by applying one  
$\Delta^+$~operation and two NNI operations
as illustrated in Figure~\ref{fig:ring}.

\def \eleSize{\fontsize{20}{20.4}}
\begin{figure*}[h]
    \begin{center}
    {\resizebox{0.95\columnwidth}{!}{   
{\includegraphics{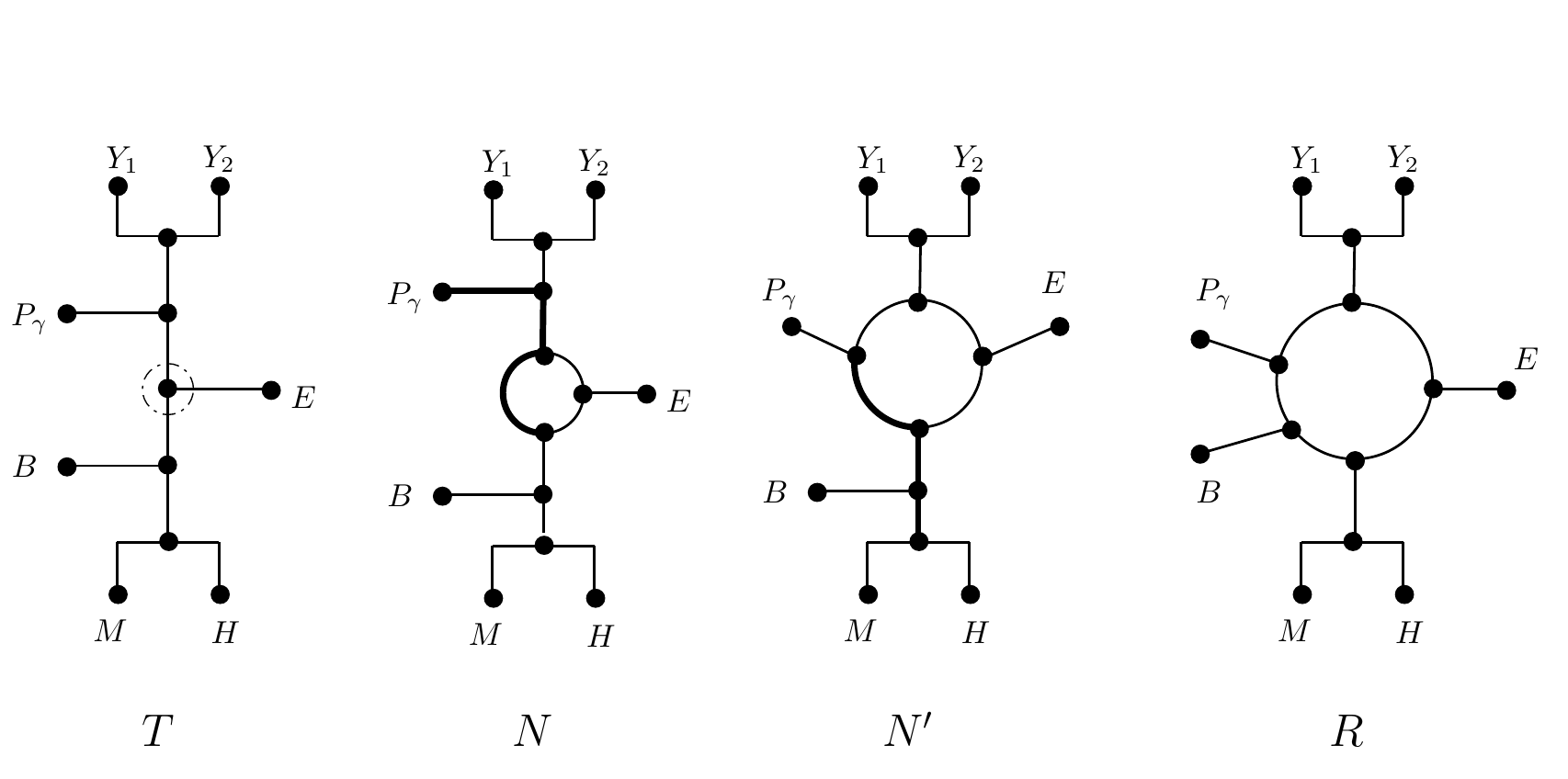}}
}  }
  \end{center}
\caption{Transforming a phylogenetic 
tree into a network using
NNI and \tc~operations. The tree
$T$ appeared in \citet[Figure 1]{riv-04}, and $R$ 
is the same network as in Figure~\ref{fig:nni:real}.
Network $N$ is obtained from $T$ by
a \tcp~operation on the vertex highlighted in 
the dash circle. Networks $N'$ and $R$ are obtained
from $T$ and $N'$, respectively, by applying
an NNI operation to the path highlighted in bold.  
}
\label{fig:ring}
\end{figure*}

We now summarize the rest of the paper. After 
presenting some preliminaries concerning graphs in the next section, 
we show that it is possible to construct 
any network starting either from a tree or from a rather simple
graph with just two vertices (called a generator). 
In Section~\ref{sec:transforming}, we then present the two main results stated above, which we prove
using some new connections between networks and the theory
of so-called cubic graphs that we discovered as part of this work.
Using our results, in Section~\ref{sec:space} we also define a generalization of
tree space, and see how this can be used
to give some new metrics that can be used to compare networks. 
We conclude in Section~\ref{sec:discussion} with a brief discussion of  
future directions and open problems, including
some potential consequences for network searching algorithms.

\section{Preliminaries}
\label{sec:preliminaries}

We begin with some technical preliminaries concerning
graphs, which give the underlying 
structure for phylogenetic networks. 
We assume throughout the paper that $X$ is a set
of species or taxa with three or more elements.

A {\em multigraph} $G$ is an ordered pair $G=(V,E)$, with 
$V$ a set of vertices and $E$ a multiset of edges 
each being an unordered
pair of elements in $V$. 
Note that we assume that multigraphs  
do not contain loops (i.e. edges that connect vertices
to themselves) unless otherwise stated, but they may contain
{\em parallel edges} (i.e. edges with multiplicity greater than 1). 
If $G$ contains no parallel edges we shall call
it a {\em graph}. We also assume that all multigraphs are 
connected. For example, 
in Figure~\ref{fig:ops}(i) we picture a connected  multigraph
with two vertices and three parallel edges.

Assume for the reminder of this section that $G$ is a multigraph.
An edge whose removal disconnects $G$ is called a {\em bridge}
of $G$; if $G$ contains no such edge, then $G$ is called {\em bridgeless}.
A {\em blob} in $G$ is a maximal connected induced submultigraph
of $G$ that does not contain a bridge.
For example, the cycle in the network $R$ in Figure~\ref{fig:nni:real} 
is a blob and so is
the subgraph of the network in Figure~\ref{fig:generator:real}(i)
with all labelled vertices
and their adjacent edges removed. 
A triangle in $G$ is called  {\em isolated} if it does not
share an edge with any other triangle
in $G$, in which case the triangle is a blob of $G$.
We say that $G$ is {\em level-$k$}, $k \ge 0$ an integer,
if a tree can be obtained from $G$ by removing at most
$k$ edges from each blob of $G$. In this case, each blob
of $G$ is at most level-$l$, $0 \le l \le k$. 

The proofs of our main results are closely 
related to some theorems in the theory of
cubic graphs. We say that $G$ is {\em cubic} 
if every vertex in $G$ has degree three. 
%
Note that a cubic multigraph necessarily has an even 
number of vertices. 
Following \citet{bri-clee-13}, we call a 
multigraph $G$ a {\em 1-3 graph} if
all vertices in $G$ have degree either 1 or 3.
These graphs naturally arise when studying phylogenetic trees
as every vertex in such a tree has degree 1 or 3.
Note that in a cubic or a 1-3 graph an 
isolated triangle cannot share a vertex with any
other triangle in that graph and
any parallel edges must have multiplicity two 
except in the (necessarily unique) level-2 cubic graph.

Finally we introduce the concept a generator since, as
we shall see later, they can be viewed as the 
building blocks of networks.
A {\em generator} is either (i) an element of $X$,
(ii) a loop with a single vertex, or (iii) a 
bridgeless, cubic multigraph. 
Note that this is an unrooted equivalent for the definition
of a generator for directed networks presented in \citet{gam-09}.
In case of (i) and (ii),
they are the (unique) generator with level-0 and level-1, respectively.
Also, for any integer $k\geq 2$, 
a generator is level-$k$ if
and only if it has $2(k-1)$ vertices. 


\section{Constructing networks}\label{sec:constructing}

Our first results concern the construction of 
networks by using simple operations. We shall first show 
that is it possible to construct any network starting from a 
phylogenetic tree using just NNI and \tc~operations. 
This is a key step in proving that we can transform any
network into any other using just NNI and \tc~operations.
Note that the T-REX algorithm mentioned in the 
introduction~\citep{makarenkov2001t} also builds 
networks from trees, but it does this 
somewhat differently by adding in new edges to 
a pre-computed tree that optimize some function. 

Although somewhat
technical, we include a proof for the first construction result
since it nicely illustrates the type of considerations
that are required to show that we can transform
one network into another using only NNI and \tc~operations.
For an illustration of this theorem see Figure~\ref{fig:ring}
(which we consider in reverse!).

\begin{theorem}\label{thm:construct}
Suppose that $N$ is a graph such that $X$ is contained in its vertex set.
Then $N$ is a  network on $X$
if and only if there is a phylogenetic tree $T$ on $X$ and 
a finite sequence of NNI and $\tcp$ operations that can be applied one
after the other, starting with $T$, to obtain $N$. 
\end{theorem}

\begin{proof}
If $T$ is a phylogenetic tree on $X$ and $N$ is
obtained from $T$ by a sequence of  NNI and \tcp~operations, then
it is straightforward to see that $N$ is a network on $X$
since applying each of these operations to a 
network results in a network.

Conversely, since each of the $\tcm$ and NNI operations 
preserve the leaf set of a network,
it suffices to show that we can reduce any 1-3 graph
to a binary tree by using a sequence of these two 
operations.

 \def \eleSize{\fontsize{20}{20.4}}
\begin{figure*}[h!]
    \begin{center}
    {\resizebox{0.95\columnwidth}{!}{   
{\includegraphics{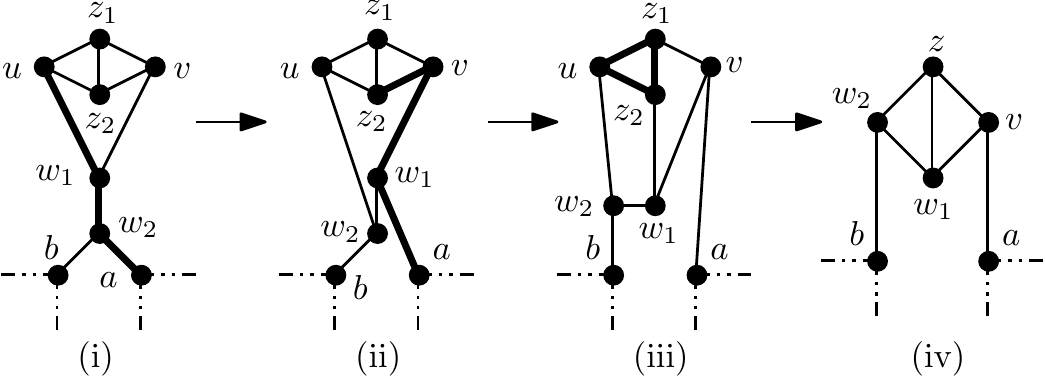}}}
 }
  \end{center}
\caption{Case II, 2.1 in Theorem~\ref{thm:construct}.
(i) and (ii): The edges in bold indicate the path
to which an NNI operation is applied. 
(iii) A $\tcm$ operation is applied to the vertices
of the triangle whose edges are given in bold. (iv) Graph obtained from
(iii); the vertex labeled $z$
is the newly generated vertex.}
\label{fig:movedown1}
\end{figure*}

To prove this, we adapt the proof
presented in \citet[p.~94-95]{bat-81}
for showing that every connected cubic graph
different from a 4-clique  (i.e., a graph with 
4 vertices and all of the 6 possible edges) can be reduced (in terms
of number of vertices) using the rules P1 and P7 
considered in that paper. 

Suppose that $N$ is a 1-3 graph which is not a tree.
Let $C$ be a shortest cycle in $N$.

\noindent {\bf Case I:} $C$ is not a triangle. 

In this case (and as exemplified by performing the NNI 
operation 
on network $R''$ that reverses the one that transforms $R'$
to $R''$ in Figure~\ref{fig:nni:real}), 
we can reduce the number of vertices of $C$
by 1, by applying an NNI operation to a path $v_1,v_2,v_3,v_4$
where $v_1$ is not a vertex in $C$ (and hence $\{v_1,v_2\}$ is not an edge in $C$) and $\{v_2,v_3\}$ and $\{v_3,v_4\}$ are two edges in $C$. Note that this
is possible since $\{v_1,v_3\}$ and  $\{v_2,v_4\}$
cannot be edges in $N$ as this would contradict the 
minimality of $C$. We can therefore repeat this
application of the NNI operation until we are in Case II.

\noindent {\bf Case II:} $C$ is a triangle. 
  
\noindent {\bf Case II, 1:} $C$ is isolated, and so we can
apply the $\Delta^-$ operation.

\noindent {\bf Case II, 2:} $C$ has an edge in common with
another triangle. 

Let $u,v$ be the two (necessarily distinct) vertices in 
the two triangles that are not contained in the common edge. Let $z_1$
and $z_2$ denote the vertices adjacent with that edge.

\def \eleSize{\fontsize{20}{20.4}}
\begin{figure*}[h!]
    \begin{center}
    {\resizebox{0.7\columnwidth}{!}{   
{\includegraphics{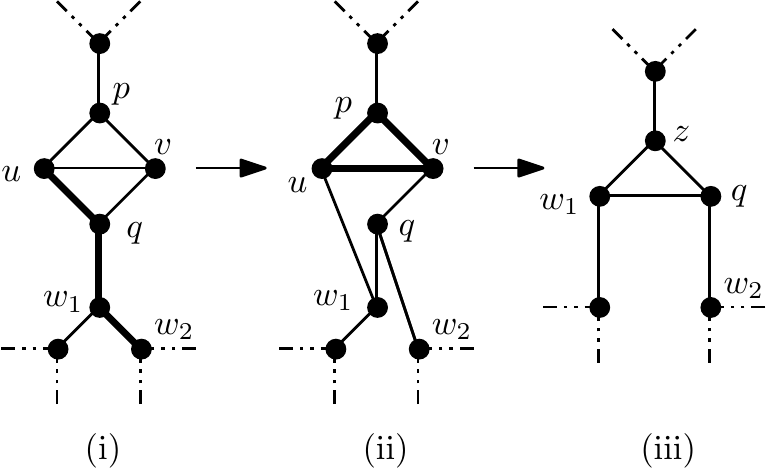}}}
 }
  \end{center}
\caption{Case II, 2.2 in Theorem~\ref{thm:construct}.
(i) The edges in bold indicate the path to which
an NNI operation is applied. (ii) A $\tcm$ operation 
is applied to the vertices
of the triangle whose edges given in bold. (iii) The 
graph obtained from (ii);  the vertex labeled $z$
denotes again the newly generated vertex.}
\label{fig:movedown2}
\end{figure*}

Note that $\{u,v\}$ is not an edge in $N$ since otherwise $N$
would be isomorphic to a 4-clique, which is not possible as
$G$ has at least three degree 1 vertices (since $|X| \ge 3$).

We distinguish now between two cases: $u,v$ have a 
neighbor distinct from $z_1$ and $z_2$ in common, or not.

\noindent {\bf Case II, 2.1:}
Let $w_1\not\in \{z_1,z_2\}$ be the vertex which is a common neighbor of
$u$ and $v$. Note that the vertex $w_2$ adjacent to $w_1$ 
that is not $u$ or $v$ must have degree 3, otherwise
$N$ would have only one vertex with degree 1 as it is a 1-3 graph which
is impossible. We can
now perform the operations in Figure~\ref{fig:movedown1} (just as
in \citet[Case II/2.1]{bat-81}).

\noindent {\bf Case II, 2.2:}
Let $p$ and $q$ be the vertices that are adjacent to
$u$ and $v$, respectively. Then at least one of $p$ or $q$ 
must have degree 3, otherwise $N$ would have only
two vertices with degree 1 implying that $X$ has two elements
 which is impossible. 
Without loss of generality, we many
assume that the degree of $q$ is 3. Now, we can reduce $N$ 
as in Figure~\ref{fig:movedown2} (just as
in \citet[Case II/2.2]{bat-81}).
\end{proof}



We now describe a second approach to constructing networks 
which is based on generators. Although we do not require this
result later on, it is of independent interest and
it also illustrates how other operations could
be potentially used for transforming networks.

Before proceeding, as an illustration of the concept 
of generators, in
Figure~\ref{fig:generator:real}(i) 
we picture a network that was referred to as 
the ``rings of life'' in \citep{ls-13}. The authors
interpreted the bottom right ring as a central system 
that contains the root of life and connects the two other rings.
As we can see this network has a more complicated
internal structure than the ring of life network $R$ 
in Figure~\ref{fig:nni:real}. This structure is further revealed if
we remove all the pendant labelled edges and suppress  
any remaining degree 2 vertices. In this case 
we obtain the level-3 generator which is pictured
in the far right of Figure~\ref{fig:generator:real}.

\def \eleSize{\fontsize{20}{20.4}}
\begin{figure*}[h!]
    \begin{center}
    {\resizebox{0.8\columnwidth}{!}{   
{\includegraphics{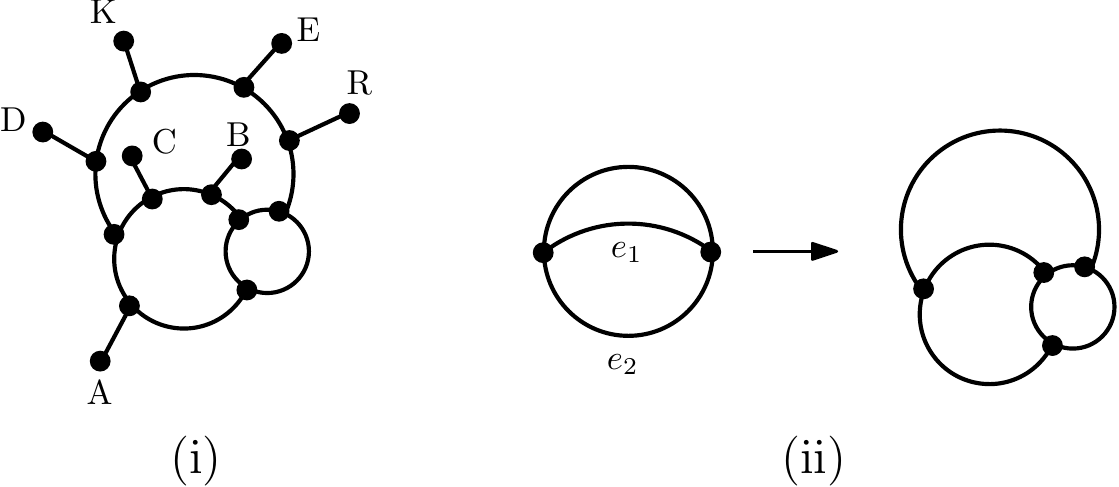}}}
 }
  \end{center}
\caption{(i) The rings of 
life from~\cite{ls-13} on eukaryotes (K), 
double-membrane (Gram-negative) 
prokaryotes (D), eocytes (E), the {\em Euryarchaeota} (R), 
{\em Actinobacteria} (A), and the {\em Firmicutes} 
(the {\em Clostridia}, C, and the {\em Bacilli}, B).  
Note that by removing all pendant edges in this 
network and suppressing degree 2 vertices we obtain
a level-3 generator.
(ii) The level-3 generator underlying the network in (i) can be 
constructed by applying one $H$ operation 
to the edges $e_1$ and $e_2$ of the theta graph on the left of 
Figure~\ref{fig:generator:real}(ii).    
\label{fig:generator:real}}
\end{figure*}

We now explain how level-$k$ generators, $k \ge 2$, 
can be constructed starting with the multigraph 
in Figure~\ref{fig:ops}(i) using just the two simple operations  
in Figure~\ref{fig:ops}(ii) which we 
call the {\em $H$} and {\em $L$ operations}, respectively.
We illustrate the $H$ operation
in Figure~\ref{fig:generator:real}(ii), where it is
used to build the level-3 generator underlying the 
rings of life network pictured in 
Figure~\ref{fig:generator:real}(i). 
This could be useful in exhaustively exploring
alternative network hypotheses such as
the rings of life example in  Figure~\ref{fig:generator:real}.
Note that the 
$H$ operation, which connects two edges by adding an 
additional edge, 
was introduced in 1891 
by J.~de Vries \citep[see, e.g.][]{bri-11}, and the $L$ 
operation, which is applied to one edge,
was considered 
in \citet[Figure 4]{bri-clee-13} 
where it was called operation 3.

\def \eleSize{\fontsize{20}{20.4}}
\begin{figure*}[h!]
    \begin{center}
    {\resizebox{0.7\columnwidth}{!}{   
{\includegraphics{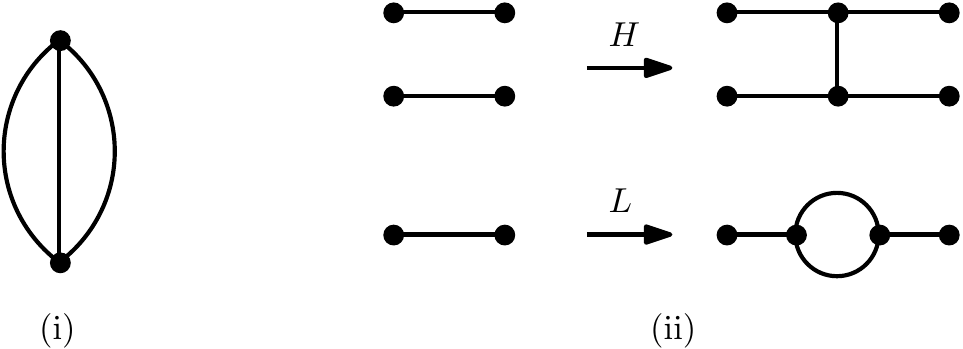}}
}  }
  \end{center}
\caption{(i)  The (necessarily unique) cubic multigraph
with two vertices (called theta graph in \citep[Figure 5]{bri-clee-13}).
(ii) The $H$ and $L$ operations on 
cubic multigraphs. Note that the two edges on which $H$ 
operates are also allowed to have one or two vertices in common (e.g.
they could be parallel).}
\label{fig:ops}     
\end{figure*}

The following result implies
that any level-$k$ generator, $k \ge 2$, can be constructed 
using just $H$ and $L$ operations. It  can be regarded
as an analogue of the approach given in \citet[Section 2]{gam-09} for
constructing generators for rooted networks. 
As we will not require this result later, we give its proof in Appendix A.

\begin{theorem}\label{thm:generators}
Suppose that $G$ is a multigraph and $k\geq 2$. 
Then $G$ is a level-$k$ generator
if and only if $G$ can be constructed from the 
theta graph in Figure~\ref{fig:ops}(ii) via performing a sequence
of $(k-2)$ $L$ or $H$ operations.
\end{theorem}

We remark that using this result it is in principle possible to 
explicitly construct (and therefore directly count) 
all level-$k$ generators.
For example, it can be checked in this way
that there is one level-1 generator, 
two level-3 generators, and six level-4 generators.
However, this could probably be done more
efficiently using the results in \citet{bri-clee-13}. 


We now briefly explain how generators can be
used to construct networks; full details are given in  
Appendix B.  The construction is based on the 
fact that if $N$ is a network  and $G$ is a blob
in $N$, then $G$ is either a level-0 generator (i.\,e.\,a vertex), 
or the graph obtained from $G$ by suppressing 
all degree 2 vertices is a level-$k$ generator, $k\ge 1$
(see Lemma~\ref{lem:generator}). It follows that we can 
build any network $N$ from the collection of generators
corresponding to the blobs of $N$ by gluing together
the generators one-by-one along new edges onto a 
growing network until
we obtain $N$. We illustrate this process in
 Figure~\ref{fig:generator}.

\def \eleSize{\fontsize{20}{20.4}}
\begin{figure*}[h!]
    \begin{center}
    {\resizebox{0.9\columnwidth}{!}{   
{\includegraphics{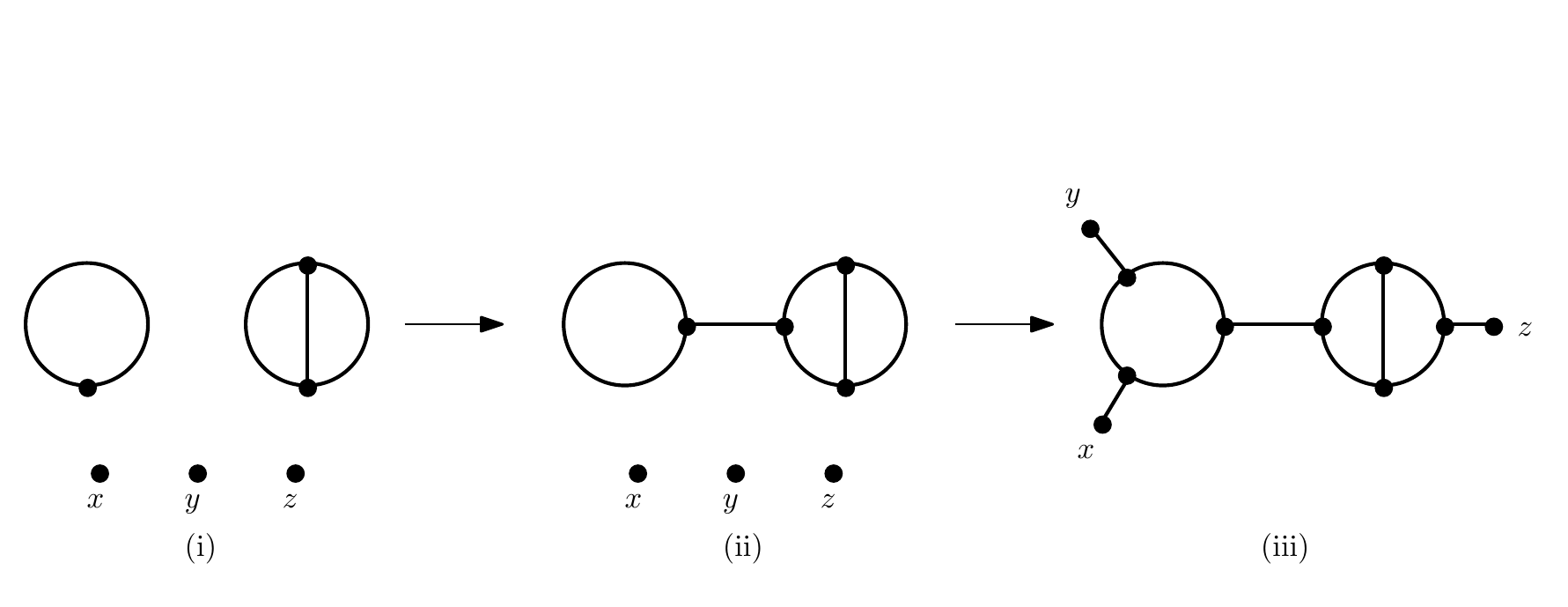}}
}  }
  \end{center}
\caption{(i) A collection of generators. (ii)  Top: a cubic multigraph
obtained by joining the level-1 and level-2 generators in (i).
(iii) A phylogenetic network obtained by 
attaching the level-0 generators to the cubic multigraph in (ii).}
\label{fig:generator}     
\end{figure*}

\section{Transforming phylogenetic networks}
\label{sec:transforming}

We now turn our attention to
the other main results of this paper. 
In the first result we see that somewhat surprisingly 
we can transform any  network to any other
network with the same number of vertices using 
only NNI operations. For an illustration 
of this result, see Figure~\ref{fig:nni:real}.

\begin{theorem}\label{transform2}
Suppose that $N,N'$ are two networks on $X$ with the same number of 
vertices. Then $N$ can be transformed 
into $N'$ by a finite sequence of NNI operations.  
\end{theorem}


\begin{proof}
Suppose that $N,N'$ are networks on $X$. We put $N \sim N'$
if, starting with $N$, there is a sequence
of NNI operations (possibly of length 0)
that can be applied  to obtain $N'$.
Note that $\sim$ is an equivalence 
relation on the set $\cN(X)$ of all networks on $X$.

Now, suppose $N,N'$ are networks on $X$ with vertex sets
$V(N)$ and $V(N')$, respectively, such that
with $|V(N)|=|V(N')|$.
The proof that $N \sim N'$ must hold is
almost identical to that in \citet[Theorem II]{tsu-96}.
However, we shall  present a sketch proof for
the sake of completeness.

As in the proof of \citet[Theorem II]{tsu-96}
we use three observations which can be proven 
in essentially the same way as
the analogous results for cubic graphs 
as indicated in brackets.

\noindent 
{\bf Observation 1 \citep[cf.][Lemma~4.2]{tsu-96}:}
If $N$ is a network on 
$X$ that is not a phylogenetic tree, then $N \sim N'$ for $N'$ 
a network on $X$ which contains an isolated triangle.

It is straightforward to check that 
Observation~1 holds by employing
an argument similar to that used in 
the proof of Theorem~\ref{thm:construct} Case I
(note that as in that theorem 
we require $|X| \ge 3$ for this to hold).
The argument presented in \citet{tsu-96} 
is similar but no reference is given 
to \citet{bat-81}.

\noindent {\bf Observation 2  \citep[cf.][Lemma~4.3]{tsu-96}:}
If $N$ is a  network on 
$X$, and $u,v \in (V(N)-X)$, then 
$N^u \sim N^v$, for the networks $N^u$ and 
$N^v$ obtained from $N$ by performing a \tcp~operation at
$u$ and $v$, respectively.

\def \eleSize{\fontsize{20}{20.4}}
\begin{figure*}[h!]
    \begin{center}
    {\resizebox{0.9\columnwidth}{!}{   
{\includegraphics{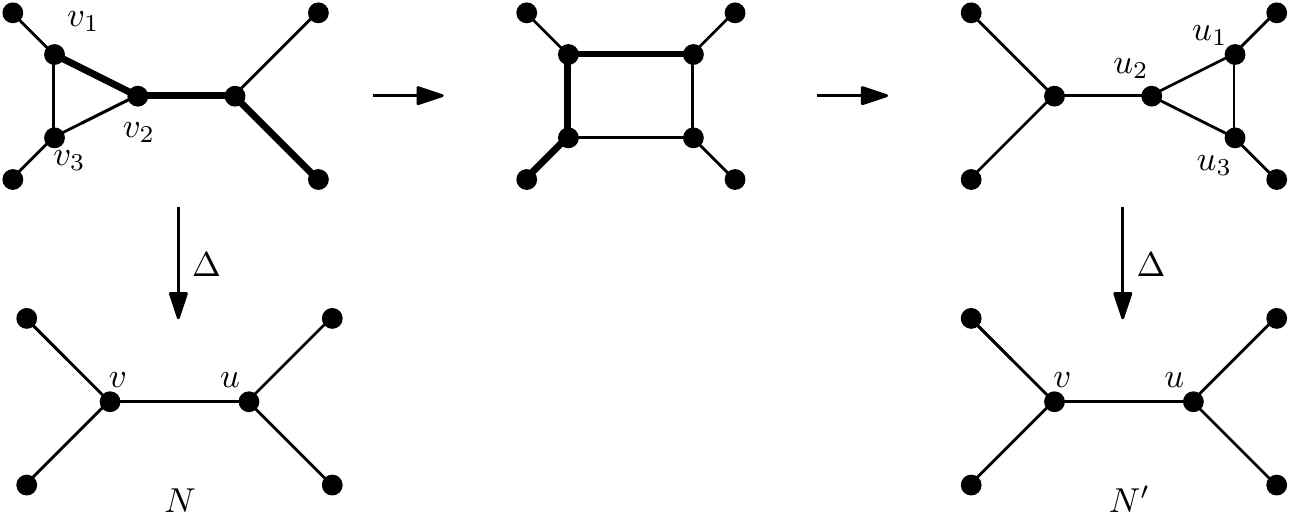}}}
 }
  \end{center}
\caption{An illustration of the proof of Observation 2 in Theorem~\ref{transform2}. The edges in 
bold indicate the paths to which NNI operations are applied.
The \tc~operations are applied to the triangles with vertices 
$v_1$, $v_2$, $v_3$ and $u_1$, $u_2$, $u_3$, respectively.}
\label{fig:obs2}
\end{figure*}

To see that Observation 2 holds, it suffices to assume that $u$
is adjacent to $v$; in this situation the
operation indicated in Figure~\ref{fig:obs2} 
\citep[adapted from][Figure 12]{tsu-96} 
can be applied to see that $N^u \sim N^v$ holds.

\noindent {\bf Observation 3  \citep[cf.][Lemma~4.4]{tsu-96}:}
Suppose that $M$ is a network on 
$X$ with an isolated triangle, and $N$ is obtained by 
performing a \tcm~operation on the vertices 
of this triangle. If $N'$ is 
obtained by performing one NNI operation on $N$, then there
is a network $M'$ on $X$ with $M' \sim M$, 
such that $N'$ can be obtained by performing a \tcm~operation on $M'$.

 \def \eleSize{\fontsize{20}{20.4}}
\begin{figure*}[h!]
    \begin{center}
    {\resizebox{0.9\columnwidth}{!}{   
{\includegraphics{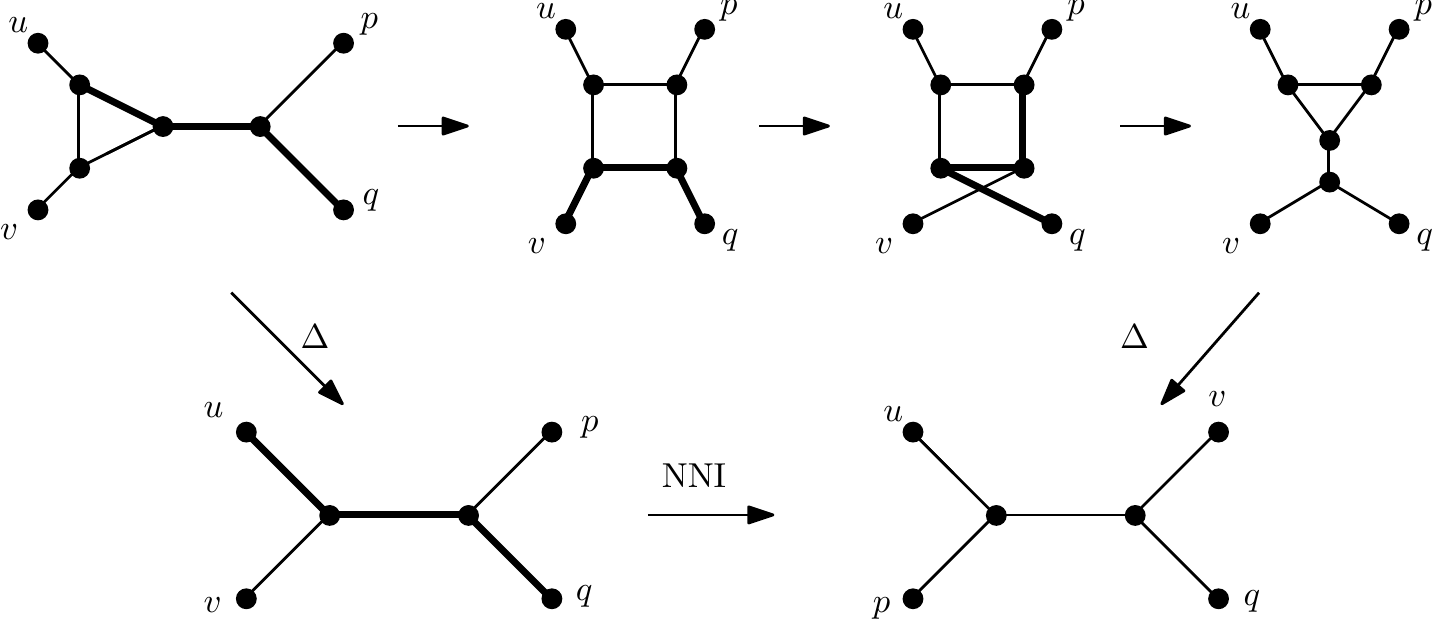}}}
 }
  \end{center}
\caption{An illustration of the proof of Observation 3 in Theorem~\ref{transform2}. The same conventions
as for Figure~\ref{fig:obs2} apply.}
\label{fig:obs3}
\end{figure*}

To see that Observation 3 holds, assume that 
the NNI operation performed on $N$ to obtain $N'$ 
is on the path $v_1,v_2,v_3,v_4$, and that $w$ is the 
vertex in $N$ obtained by applying the $\Delta^-$ operation
to the vertices of an isolated triangle in $M$. 
It is straightforward to see that Observation 3 holds
in case $w \neq v_2,v_3$. In case $w=v_2$ (or $w=v_3$) the 
way to obtain the network $M'$ is indicated in 
Figure~\ref{fig:obs3} \citep[adapted from][Figure 14]{tsu-96}.

Using these three observations we now prove that $N \sim N'$, using 
induction on $m=|V(N)|=|V(N')|$. We shall use the fact that 
if $N$ is a network on $X$, then $|V(N)| \ge 2|X|-2$, with equality holding 
if and only if $N$ is a phylogenetic tree on $X$. 

First note that if $|V(N)| = 2|X|-2$, then both $N$ and $N'$
are trees, and so $N \sim N'$ \citep{rob-71}.
So, suppose that $M \sim M'$ holds for all networks $M$ and
$M'$ on $X$ with
$|V(M)|=|V(M')|=m \ge 2|X|-2$ holding for some $m$. Assume that $N,N'$
are two networks on $X$ such that $|V(N)|=|V(N')|=m+2$. 

By Observation~1, we can assume that $N,N'$ have isolated
triangles $\Gamma$ and $\Gamma'$, respectively. Let $P,P'$ be the
networks obtained by performing an $\Delta^-$~operation on
$N$ and $N'$ on the vertices of  the 
triangles $\Gamma$ and $\Gamma'$, respectively. 

By induction, there exists a sequence $P=P_0,P_1,\dots,P_q=P'$
with network $P_{i+1}$ obtained from $P_i$ by an NNI operation, 
$i=0,1,\dots,q-1$.

Now, by Observation~3, it follows that 
there is a sequence $N_0,N_1,\dots,N_q$
so that $N_0 = N$, $N_i$ has an isolated triangle
$\Gamma_i$ with $\Gamma_0=\Gamma$, $N_i \sim N_{i+1}$, 
and $P_i$ is the network obtained by performing
an $\Delta^-$~operation on the vertices of the triangle $\Gamma_i$ in $N_i$.
In particular, $N_q \sim N$ and the
$\Delta^-$~operation applied to the vertices of the triangle 
$\Gamma_q$ in $N_q$
gives the network $P_q=P'$.

But then there must exist vertices $u,v$ in $V(P_q)-X$ so
that the networks $N_q$ and $N'$ can be obtained from $P_q$ 
by performing a \tcp~operation at $u$ and $v$,
respectively. By Observation~2, $N \sim N'$ follows, as required.
\end{proof}

Using the last theorem we can now prove our second
main result, namely that we can transform
any network on $X$ to any other network on $X$ by
just using NNI and \tc~operations.

\begin{theorem}\label{transform1}
Suppose that $N$ and $N'$ are networks on $X$. Then, 
starting with $N$, there exists 
a finite sequence of NNI and $\tc$   
operations that can be applied to obtain $N'$. 
\end{theorem}

\begin{proof}
By Theorem~\ref{thm:construct}, there are phylogenetic 
trees $T,T'$ on $X$ that can be obtained from networks $N,N'$
on $X$, respectively, by applying a sequence of NNI and  
\tc~operations. 
Therefore the theorem follows from the 
fact that there is a sequence of NNI operations that can be 
applied, starting with $T$, to obtain $T'$, as shown in \citet{rob-71}.
\end{proof}

\section{Network spaces}
\label{sec:space}

For phylogenetic trees, NNI operations lead naturally
to the concept of {\em tree space}~\citep[see e.g.][]{bhv01}. 
In its simplest form,
for $X$ fixed, this space can be thought of as the set 
of all phylogenetic trees on $X$ with two trees
being neighbors if and only if they differ
by one NNI operation. In other words,
tree space is a graph with vertices being the
phylogenetic trees on $X$ and edges corresponding
to pairs of such trees which differ by one NNI operation.
As a consequence of the fact that any phylogenetic tree on $X$
can be transformed into any other by a sequence of
NNI operations, it follows that this graph is connected.
This is important for optimization algorithms that search tree space
using NNI operations since it implies that, 
in principle, every tree could be visited. 

Using the operations introduced in this paper, we can 
define network space in a similar manner. Indeed,
we can consider a graph with vertices consisting of 
all networks on $X$ and edges corresponding to pairs
of networks which differ by either one NNI operation
or one \tc~operation. Note that this space is
well defined since NNI and \tc~operations are 
both reversible. We illustrate part of
this network space for $X= \{a,b,c,d\}$ in
Figure~\ref{fig:space}. 

By Theorem~\ref{transform1} it 
follows that network space is connected, just
as with tree space. In fact, tree space is
actually a subgraph of network space, as it consists
of all of the level-0 networks (i.e. phylogenetic trees)
and the NNI operation is just the usual NNI operation for trees.
Moreover, Theorem~\ref{transform2} shows
that  network space has some additional structure since, for any non-negative 
integer $i$, the set $\cN_i(X)$ of all networks on $X$ 
with $2(|X|-1+i)$ vertices forms a 
``tier'' 
in the space in the sense that any pair of
networks in that tier can be joined by edges 
corresponding only to NNI operations (Figure~\ref{fig:space}).
Moreover, two tiers $\cN_i(X)$ and $\cN_j(X)$ are 
adjacent if and only if $|j-i|=1$, that is, the difference 
between the number of vertices for the networks in 
tier $\cN_i(X)$ and that in tier $\cN_j(X)$ is 
exactly two. In particular, there exists a network $N$ in $\cN_i(X)$ 
and a network $N'$ in $\cN_j(X)$ such that $N$ and $N'$ 
differ by precisely one \tc~operation.
Note that partitioning network space into 
tiers is not equivalent to
dividing up the set of networks according
to their level. This alternative way 
to partition rooted networks has been 
considered in \citet{yu2014maximum}  
and also in \citet{hub-15} where an 
alternative space of level-1 networks is defined.

\begin{figure*}[h!]
    \begin{center}
    {\resizebox{0.9\columnwidth}{!}{   
{\includegraphics{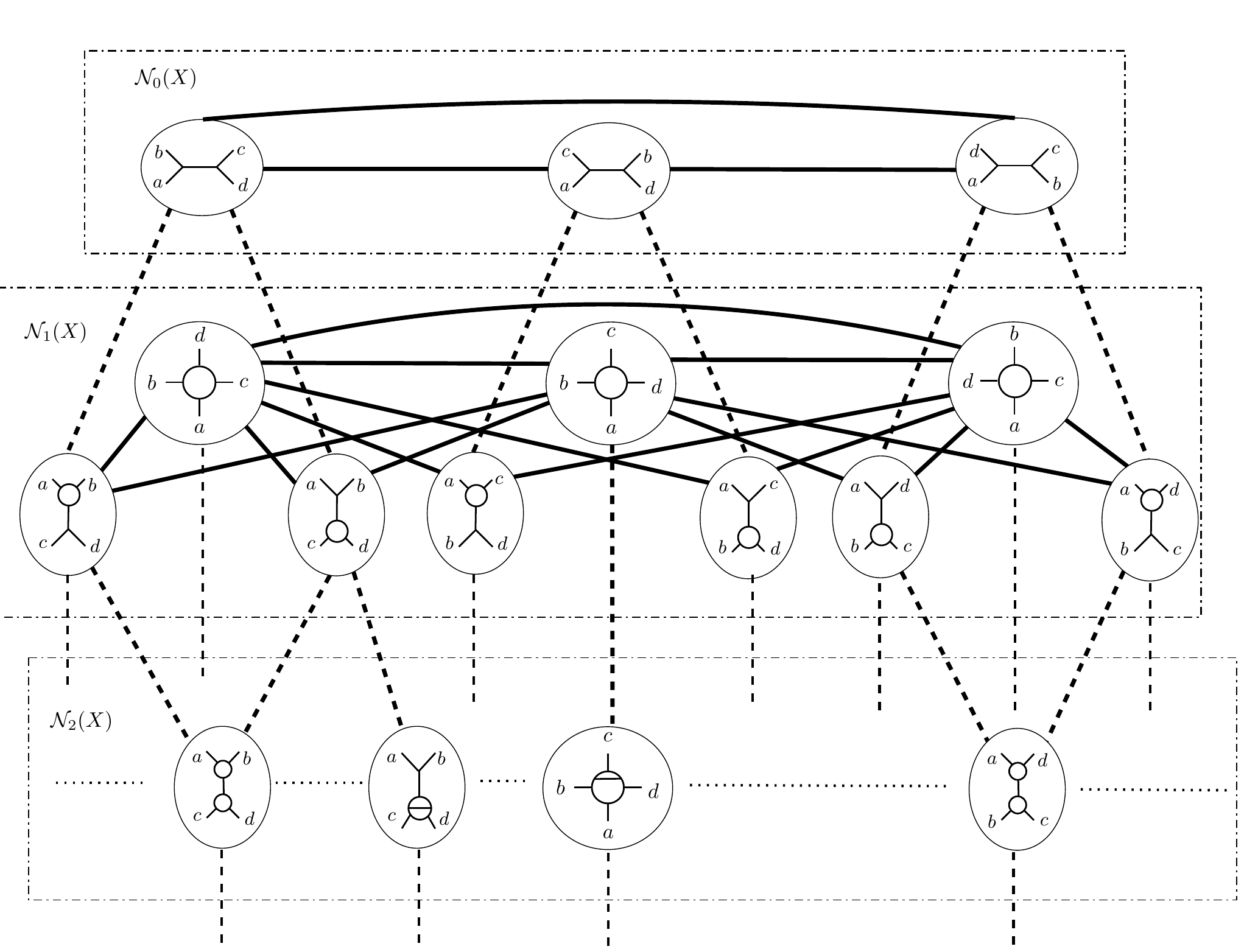}}}
 }
  \end{center}
\caption{Part of the space of phylogenetic networks on the leaf set $\{a,b,c,d\}$. Two networks are connected by a solid line if one 
can be transformed into the other by one NNI operation, and a 
dashed line if this can be done by one \tc~operation. This gives
rise to the tiers 
$\mathcal{N}_i(X) $ described in the text.
%
}
\label{fig:space}
\end{figure*}

Tree spaces have been extensively studied in the 
literature, and they have various important 
properties, many of which have direct bearing 
on the development and performance of tree search
algorithms. As an example of this, tree space
naturally gives rise to the well-known NNI metric on trees
which is the shortest path between any pair
of trees in tree space, or equivalently, the 
least number of NNI operations required to transform one tree into 
another~\citep[cf.][and the references therein]{das-97}.
As a consequence of our results we can 
define new metrics on the set $\cN(X)$ of
all networks on $X$ which arise from
network space in a similar way. 

More specifically, for $N,N' \in \cN_i(X)$ 
we define the distance $d_{\NNI}(N,N')$ between $N$ and $N'$
to be the minimal number of NNI 
operations required to transform $N$ into $N'$,
and for $N,N' \in \cN(X)$ we let the distance $d(N,N')$
between $N$ and $N'$  
be the minimal number of $\tc$ and NNI operations needed to 
transform $N$ into $N'$. 
By Theorems~\ref{transform2} and~\ref{transform1} we 
immediately obtain the following result.

\begin{theorem}\label{metrics}
The distance $d$ is a metric on the set $\cN(X)$. 
In addition, for any non-negative integer $i$, the 
distance $d_{\NNI}$ is a metric on the set $\cN_i(X)$.
\end{theorem}

It would be interesting to better understand 
properties of these metrics such as 
ways to compute or find bounds on the size of the
distance between any pair of networks. Note that the
NNI distance between a pair of phylogenetic trees is 
NP-hard to compute in general \citep{das-97},
and so as tree space is contained in network space 
it is unlikely that there will 
be a polynomial time algorithm to compute
the metric $d$ in general. However, understanding 
properties of $d$ could yield clues
as to how search algorithms on networks may work,
as illustrated by the recent work on Ricci-Ollivier 
curvature of tree space based on the {\em subtree
prune and regraft (SPR)} operation \citep{whidden2015ricci}.

\section{Discussion}
\label{sec:discussion}

In this paper we have shown that it is possible to 
transform any network into any other network through
using a sequence of two elementary operations, one
of which is a generalized version of the NNI operation 
on trees. This allows us to consider spaces of 
networks which generalize spaces of trees, and which
could serve as search spaces for network construction algorithms. 

The NNI and \tc~operations are well-known in the theory of 
cubic graphs, where the NNI operation 
has been called a {\em Whitehead move} \citep{CU} or 
a {\em slide transformation} \citep{tsu-96},
and the \tc~operation is known as the {\em $P1^{-}$ rule} \citep{bat-81} 
and {\em $\Delta$-reduction} \citep{tsu-96}. We believe that links 
between the theories of networks and cubic graphs 
could be potentially of great interest since 
by using theoretical tools for cubic graphs and associated
structures, it should be possible to obtain new approaches
to understand and construct phylogenetic networks. 
For example, spaces of cubic graphs appear in
algebraic geometry \citep{bhv01} and techniques from
that area have proven useful in phylogenetics (e.g.
in understanding phylogenetic invariants~\citep{allman2003phylogenetic}). 
As another illustration of these possibilities, we 
note that properties of cubic graphs were 
recently used by \citet{M15} to help 
count rooted phylogenetic networks.

Although we have focused on unrooted networks (in which each 
vertex has degree 1 or 3), our results serve
to illustrate the possibilities for developing similar 
results for spaces
of more complicated networks. For example, it 
would be very interesting to see how
our results could be extended to rooted networks. 
Some results in this direction are presented in
\citep{rad-12,yu2014maximum}. Our
previous work on spaces of level-1 rooted networks~\citep{hub-15}  
suggests
that this will probably be somewhat more technical.
However, as algorithms are now being developed for
performing likelihood and Bayesian searches on
rooted networks \citep{rad-12,yu2014maximum}, it seems to be well
worth while trying to 
understand general properties of rooted network spaces.

In light of our results, it is also of interest 
to consider how network optimization searches could be performed
based on a combination of either fixing network size (e.g. 
for unrooted networks
using only NNI operations) and/or level (cf. \citet[p. 16450]{yu2014maximum}
for a search strategy on rooted networks 
which fixes level using non-local operations).
In regards to this, we could also consider other types of operations
such as generalizations of the SPR and 
{\em tree bisection and reconnection (TBR)} operations
on trees \citep{all-ste-01a}. Note that, in contrast to NNI 
operations,
these generalizations will not necessarily be local, and that there
are other operations already defined for 
cubic graphs~\citep[cf. e.g.][]{bat-81} which 
could also yield interesting network operations.
Of course, there are pros and cons in deciding which 
operations to consider. 
For example, for tree space, NNI neighborhoods are
smaller than SPR or TBR neighborhoods~\citep[cf. e.g.][]{hw-13} 
which, in practical 
terms, means that searches based on the latter operations
can be more computationally expensive. However, SPR and
TBR operations can potentially yield algorithms that are less likely
to become trapped in local optima. Thus
it would be interesting to better understand 
neighborhoods in network spaces, and to use
their structure to develop new search algorithms.

Another technical issue in this respect 
is that network spaces are in general infinite, as 
it is always possible to 
increase the level of any network by adding in 
vertices and edges (e.g. using the \tc~operation). 
Hence it could be useful to develop
ways to determine bounds for the level of {\em any} network
for any given data set (maybe using concepts such as the 
Bayesian information criterion (BIC)~\citep{yu2014maximum}), and then 
only search spaces of networks with bounded level 
for that data set. This would at least give a finite
search space, such as the space of level-1 networks
that we defined in \citet{hub-15}.


\appendix

\section{Proof of Theorem~\ref{thm:generators}} 

\begin{proof} 
Suppose that $G$ is a multigraph and that $k\geq 2$.
If, starting with the theta graph, $G$ is constructed via a sequence
of $(k-2)$  $L$ or $H$ operations 
 then it is straightforward to
see that $G$ is a level-$k$ generator.

Conversely, suppose that $G$ is a level-$k$ generator.
We proceed by using induction on the level of $G$.
If $k=2$, then $G$ is the theta graph, and so the base case 
clearly holds. So, suppose that any level-$l$ generator, $2 \le l < k$,
can be constructed using $(l-2)$ $H$ or $L$ operations.

If
$G$ contains a 
parallel edge connecting two distinct vertices $u$ and $v$, say,
then it must contain the subgraph pictured in Figure~\ref{fig:configs}(i). 
But then by performing a reverse $L$ operation on this pair of 
parallel edges (i.e., deleting one of the two parallel 
edges and suppressing the resulting degree two vertices) 
we obtain a bridgeless, cubic multigraph $G'$ (as
$G$ is a level-$k$ generator for some $k>2$, and 
hence a bridgeless, cubic multigraph). Note that this operation 
decreases the level of $G$ by one and so,
by induction, $G'$ can be constructed using $k-3$
$H$ or $L$ operations. By induction the theorem therefore holds
in this case.

\begin{figure*}[h!]
    \begin{center}
    {\resizebox{0.7\columnwidth}{!}{   
{\includegraphics{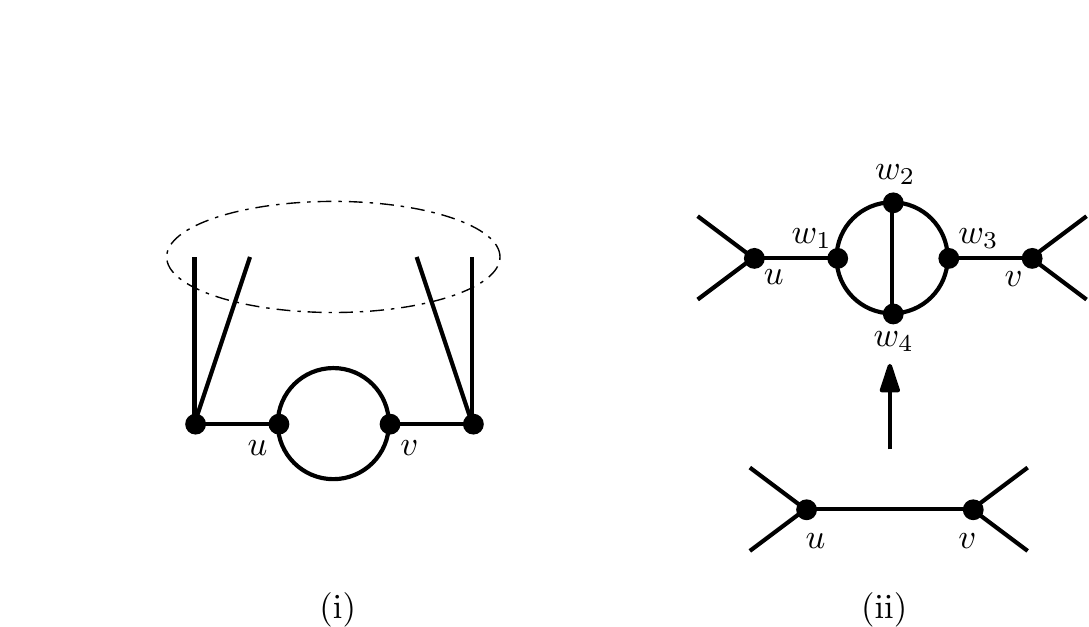}}}
 }
  \end{center}
\caption{Two configurations considered in the proof of 
Theorem~\ref{thm:generators}. The  vertices $w_1,\ldots, w_4$
and their adjacent edges of the form $\{w_i,w_j\}$, $i,j\in\{1,\ldots,4\}$ 
are added by the operation mentioned in the text.}
\label{fig:configs}    
\end{figure*}

So, suppose $G$ does not contain any parallel edges, i.e. 
$G$ is a graph. Then by the main theorem in \citet[part (T8)]{bat-81}, 
it follows that either (a) $G$ is a 4-clique, 
or (b)  there is a
bridgeless, cubic multigraph $G'$ so that $G$ can be obtained
from $G'$ by performing either an $H$ operation or
an operation of the form in Figure~\ref{fig:configs}(ii). 
If case (a) holds, then we can perform a single reverse $H$
operation on $G$ to obtain the theta graph, and the 
theorem is then easily seen to hold.
In case (b) holds and $G$ can be obtained
from $G'$ by performing an $H$ operation, the theorem 
also follows by induction.

So, suppose that (b) holds and that $G$ can be obtained
from $G'$ by using an operation as in Figure~\ref{fig:configs}(ii).
But this operation is equivalent 
to first performing an $L$ operation 
and then an $H$ operation, starting with $G'$.
Since the graph $G''$ obtained
from $G'$ by performing just the $L$ operation is clearly  a bridgeless
cubic multigraph with level $k-3$, and $G$ is obtained
by performing an $H$ operation on $G''$, this case
again follows by induction.
\end{proof}

\section{Constructing networks from generators}

We show how generators can be used to construct networks.
We begin with a simple observation (see  Figure~\ref{fig:generator}).

\begin{lemma}\label{lem:generator}
Suppose that $N$ is a network and $G$ is a blob
in $N$. Then $G$ is either a level-0 generator, 
or the graph obtained from $G$ by suppressing 
all degree 2 vertices (one-by-one until this 
is no longer possible) 
is a level-$k$ generator, $k\ge 1$.
\end{lemma}
\begin{proof}
Suppose that $G$ is a blob in $N$ and that $G$ is
not a level-0 generator (i.e. a vertex). 
Then, as $G$ is a blob, every vertex in $G$ 
must have degree 2 or 3.

If every vertex in $G$ has degree 2, then $G$
must be a cycle. Therefore if we suppress all but one
degree 2 vertex of $G$  we obtain a loop
with a single vertex, i.e. the level-1 generator.

If $G$ contains some vertex with degree 3, then 
the multigraph $G'$ obtained by suppressing all
vertices in $G$ of degree 2 must be a bridgeless
cubic multigraph (since $G$ has no bridges), 
i.e. $G'$ is a level-$k$ generator, for some $k \ge 2$.
\end{proof}

In light of Lemma~\ref{lem:generator} it follows that we 
can construct any network $N$ from some collection 
of generators. To see this, let $\cC_N$ be the 
collection of generators (one for each blob in $N$)
which is obtained by subjecting each blob in
$N$ to the vertex suppression process described in 
the statement of that lemma, if necessary. 
Then it is straightforward to check that
to construct $N$ from $\cC_N$, we can 
start with some generator in $\cC_N$ and then 
continue to graft
generators $G$ in $\cC_N$ that 
have not been previously selected 
onto a growing multigraph $M$
(which possibly contains loops) until we obtain $N$ 
by applying one of the following operations.
(a) If $M$ is the level-0 generator $u$ or
the level-1 generator with vertex $u$, then
to attach a new generator $G$ to $M$ add in the edge $\{u,v\}$,
where $v$ is the unique vertex in $G$ if
$G$ is level-0 or level-1, or $v$ is a new
vertex which is added to $G$ so as to subdivide some edge in $G$.
(b) If $M$ is neither the level-0
nor the level-1 generator, then subdivide some
edge in $M$ by adding in a new vertex $u$, and then, to 
attach a new generator $G$, add in the edge $\{u,v\}$
where $v$ is as in (a).
Note that this process 
can be regarded as an unrooted analogue of Theorem 3.1 in
\citet{gam-09} which describes a
similar way to construct directed phylogenetic networks from
rooted generators.

\section*{Acknowledgments}

We thank Simone Linz for interesting discussions on spaces of phylogenetic networks, and Lars Jermiin and his colleagues
for helpful comments on an earlier version of this manuscript. We also
thank the Institute for Mathematical Sciences, Singapore,
for inviting us to a workshop in July 2015, where 
some of the topics presented in this paper were discussed.

\bigskip
\noindent
{\bf References}







\end{document}